\documentclass[11pt]{article}
\usepackage[utf8]{inputenc}
\usepackage[margin=1in]{geometry} 
\usepackage[usenames,x11names]{xcolor}

\usepackage{authblk}
\usepackage{comment}
\usepackage{amsmath,amssymb,amsthm}
\usepackage{mathtools}
\usepackage{thm-restate}

\usepackage[linesnumbered,ruled,vlined]{algorithm2e}

\usepackage[colorlinks,linkcolor=DeepSkyBlue3,citecolor=SpringGreen3]{hyperref}
\usepackage[nameinlink]{cleveref} 

\newtheorem{theorem}{Theorem}

\newtheorem{lemma}[theorem]{Lemma}

\newtheorem{proposition}[theorem]{Proposition}
\newtheorem*{remark}{Remark}
\theoremstyle{plain}

\DeclareMathOperator{\E}{\mathrm{E}}

\title{Testing Graph Properties with the Container Method \footnote{An earlier version of this article appeared in the proceedings of the 2023 IEEE 64th Annual Symposium on Foundations of Computer Science (FOCS).}}

\author{Eric Blais} 
\author{Cameron Seth}
\affil{University of Waterloo}

\begin{document}
\maketitle

\begin{abstract}
We establish nearly optimal sample complexity bounds for testing the $\rho$-clique property in the dense graph model. Specifically, we show that it is possible to distinguish graphs on $n$ vertices that have a $\rho n$-clique from graphs for which at least $\epsilon n^2$ edges must be added to form a $\rho n$-clique by sampling and inspecting a random subgraph on only $\tilde{O}(\rho^3/\epsilon^2)$ vertices. We also establish new sample complexity bounds for $\epsilon$-testing $k$-colorability. In this case, we show that a sampled subgraph on $\tilde{O}(k/\epsilon)$ vertices suffices to distinguish $k$-colorable graphs from those for which any $k$-coloring of the vertices causes at least $\epsilon n^2$ edges to be monochromatic. The new bounds for testing the $\rho$-clique and $k$-colorability properties are both obtained via new extensions of the graph container method. This method has been an effective tool for tackling various problems in graph theory and combinatorics. Our results demonstrate that it is also a powerful tool for the analysis of property testing algorithms. 
\end{abstract}

\section{Introduction}

Is it possible to test if a graph contains a large clique or if it is $k$-colorable while examining only a small fraction of that graph? These problems are two of the foundational examples in graph property testing that were originally considered in the groundbreaking work of Goldreich, Goldwasser, and Ron~\cite{goldreichPropertyTesting1998}. The problems can be specified formally in the dense graph property testing framework defined as follows.

A (simple undirected) graph $G$ on $n$ vertices is \emph{$\epsilon$-far} from a graph property $\Pi$ if we need to add or remove at least $\epsilon n^2$ edges from $G$ to obtain a graph that has the property $\Pi$. 
A \emph{canonical $\epsilon$-tester} for $\Pi$ with sample cost $s$ is a bounded-error\footnote{By \emph{bounded error} we mean there exist absolute constants $\delta_1>\delta_2$ such that if $G$ has property $\Pi,$ then the algorithm accepts with probability at least $\delta_1,$ and if $G$ is $\epsilon$-far from $\Pi,$ then the algorithm accepts with probability at most $\delta_2.$} randomized algorithm that samples a set $S$ of $s$ vertices of some unknown graph $G$, examines the induced subgraph $G[S]$, and based only on this local view of the graph can distinguish between the case where $G$ has property $\Pi$ and the case where it is $\epsilon$-far from $\Pi$.
The \emph{sample complexity} of $\Pi$, denoted $\mathcal{S}_{\Pi}(n,\epsilon)$, is the minimum value $s$ for which there is a canonical $\epsilon$-tester for $\Pi$ with sample cost $s$.

\subsection{Testing Cliques}

The \textsc{$\rho$-Clique} property is the set of all graphs on $n$ vertices that contain a clique on $\rho n$ vertices. 
The problem of testing the \textsc{$\rho$-Clique} property can be considered in both the \emph{large clique} regime, where $\rho$ is a constant, and in the \emph{small clique} regime, where $\rho := \rho(n)$ is a function of the number of vertices of the graph. In the small clique regime, it is important to note that the problem is non-trivial only when $\epsilon < \rho^2$ is also a function of $n$, since we can always add at most $\binom{\rho n}{2} < \rho^2 n^2$ edges to any graph to form a $\rho n$-clique.

What is the sample complexity $\mathcal{S}_{\rho\textsc{-Clique}}(n,\epsilon)$ for $\epsilon$-testing the \textsc{$\rho$-Clique} property? 
Goldreich, Goldwasser, and Ron~\cite{goldreichPropertyTesting1998} were the first to consider this question and showed that $\mathcal{S}_{\rho\textsc{-Clique}}(n,\epsilon) = \tilde{O}(\rho/\epsilon^4)$. 
This result has remarkable qualitative implications in both the large clique and the small clique regimes.
In the large clique regime, it shows that $\epsilon$-testing $\rho$-cliques requires only inspection of a \emph{constant}-sized subgraph when $\epsilon > 0$ is an absolute constant. 
And in the small clique regime, it shows that the sample complexity for $\epsilon$-testing $\rho$-clique property is \emph{sublinear} in $n$ for all $\rho = \omega(n^{-1/7})$ when $\epsilon = \Omega(\rho^2)$---i.e., when the tester must distinguish graphs with $\rho n$-cliques from those whose $\rho n$-subgraphs all have constant density bounded away from $1$. 

Improved bounds on the sample complexity for testing cliques were obtained by Feige, Langberg, and Schechtman~\cite{feigeCliqueTesting2004}.
With a new direct analysis of the canonical tester for testing cliques, they showed that $\mathcal{S}_{\rho\textsc{-Clique}}(n,\epsilon) = \tilde{O}(\rho^4/\epsilon^3)$.
And by considering the restricted problem of distinguishing graphs that consist of a single $\rho n$-clique from those that consist of a single $(\rho - \frac{\epsilon}{\rho})n$-clique, they showed that
$\mathcal{S}_{\rho\textsc{-Clique}}(n,\epsilon) = \tilde{\Omega}(\rho^3/\epsilon^2)$.

Our first main result is a new upper bound on the sample complexity of testing the \textsc{$\rho$-Clique} property that is nearly optimal, since it matches the lower bound of Feige, Langberg, and Schechtman up to polylogarithmic factors. 

\begin{theorem}
\label{thm:clique}
The sample complexity of the \textsc{$\rho$-Clique} property is $\mathcal{S}_{\rho\textsc{-Clique}}(n,\epsilon) = \tilde{O}(\frac{\rho^3}{\epsilon^2})$.\footnote{Here and throughout the article, we use $\tilde{O}(\cdot)$ and $\tilde{\Omega}(\cdot)$ notation to hide terms that are polylogarithmic in the argument. See \Cref{sect:testing-IS,sect:colorable} for the precise formulation of the main theorems.}
\end{theorem}

The qualitative implications of \Cref{thm:clique} are most striking in the small clique regime. In this regime, the result has a more natural representation when we reformulate it using notation from the Densest $k$-Subgraph (D$k$S) problem. (For more details on the D$k$S problem, see~\cite{Manurangsi2017} and the references therein.)

\newtheorem*{clique-thm-alt}{Theorem 1'}
\hypertarget{thm:clique-alt}{\begin{clique-thm-alt}[Alternative formulation]
For every $n$, $k := k(n) < n$, and $\delta := \delta(n) > 0$, there is a bounded-error randomized algorithm that distinguishes (i) graphs that contain a $k$-clique from (ii) graphs whose $k$-subgraphs all contain at most $(1-\delta) \binom{k}{2}$ edges by inspecting the induced subgraph of the input graph $G$ on a random set $S$ of only $|S| = O(\frac{n}{\delta^2 k} \ln^3 (\frac{n}{\delta^2 k}))$ vertices.
\end{clique-thm-alt}}

This result implies that for every constant $\delta > 0$ and every $k = \omega(\ln^3 n)$, we can distinguish graphs with a $k$-clique from graphs whose $k$-subgraphs have density at most $1-\delta$ by examining only a sublinear portion of the graph. The previous bounds in~\cite{goldreichPropertyTesting1998} and~\cite{feigeCliqueTesting2004} showed that sublinear sample complexity is achievable for constant $\delta$  when $k = \tilde{\omega}(n^{6/7})$ and $k = \tilde{\omega}(n^{1/2})$, respectively.

This result can also be viewed as a generalization of results of R{\' a}cz and Schiffer~\cite{RaczSchiffer2019} and of Huleihel, Mazumdar, and Pal~\cite{HuleihelMP2021} regarding the sample complexity of the Planted $k$-Clique Detection problem. 
The latter result shows that sampling an induced subgraph on $O(\frac{n}{k\delta} \ln {\frac nk})$ vertices suffices to distinguish (i) random graphs with edge probability $(1-\delta)$ from (ii) random graphs with edge probability $(1-\delta)$ containing a planted $k$-clique.
The former result shows a similar bound for the restricted case of $\delta = 1/2.$
When $\delta>0$ is a constant, the bound in \hyperlink{thm:clique-alt}{Theorem 1'} shows that the same sample complexity (up to logarithmic factors) suffices to solve the more general perfect completeness decision version of the Densest $k$-Subgraph problem as defined in the theorem statement.

Moreover, \hyperlink{thm:clique-alt}{Theorem 1'} also gives a nearly optimal upper bound on the \emph{query} complexity of algorithms that solve the Densest $k$-Subgraph problem in the regime where $\delta > 0$ is constant.
In this setting, the algorithm is free to query any vertex pair (instead of selecting a set $S$ of vertices and querying all the pairs of vertices in $S$ to determine $G[S]$ exactly), and is free to \emph{adaptively} select which pairs of vertices to query based on the results of its earlier queries. 
The bound in the theorem implies that there is a bounded-error algorithm that distinguishes graphs with $k$-cliques from those whose subgraphs have at most $(1-\delta)k^2$ edges and has query complexity at most $O(\frac{n^2}{\delta^4 k^2}\ln^6(\frac{n}{\delta^2 k}))$. 
However, as shown in~\cite{RaczSchiffer2019,HuleihelMP2021}, at least $\Omega(\frac{n^2}{\delta^2 k^2} \ln^2 {\frac nk})$ queries are required to solve the easier $k$-Planted Clique problem, even for adaptive algorithms.

\subsection{Testing Colorability}

The \textsc{$k$-Colorable} property is the set of all graphs on $n$ vertices that are $k$-colorable. The case where $k=2$ corresponds to bipartiteness.
For testing bipartiteness, nearly optimal bounds $\mathcal{S}_{2\textsc{-Colorable}}(n,\epsilon) = \tilde{\Theta}(1/\epsilon)$~\cite{goldreichPropertyTesting1998,AlonKrivelevich02} are known. We focus on the case where $k \ge 3$. We again consider two different regimes: the \emph{low chromatic number} regime where $k$ is a constant, and the \emph{polychromatic} regime where $k = k(n)$ is a function of $n$.
In the polychromatic regime, we note that $\epsilon$-testing the \textsc{$k$-Colorable} property is non-trivial only when $\epsilon < 1/k$ (and so also depends on $n$) since we can always eliminate at most $k \binom{n/k}{2} \le n^2/k$ edges from a graph to make it $k$-colorable.

The study of the sample complexity of the \textsc{$k$-Colorable} property has a long history that predates the definition of property testing itself.  R\"odl and Duke~\cite{RodlDuke85}, building on prior work of Bollob\'as, Erd\H os, Simonovits, and Szemer\'edi~\cite{BollobasESS78}, showed that $\mathcal{S}_{k\textsc{-Colorable}}(n,\epsilon)$ is a constant independent of $n$ when both $k$ and $\epsilon$ are constant. However, this result relies on the regularity lemma and so the bounds obtained on the sample complexity grow as a tower of height polynomial in $1/\epsilon$. Notably, these results say very little in the polychromatic regime since they give a sublinear bound on the sample complexity of testing $k(n)$-colorability only for some $k(n)$ that are iterated logarithm functions of $n$.

Goldreich, Goldwasser, and Ron~\cite{goldreichPropertyTesting1998} obtained much better bounds for the \textsc{$k$-Colorable} property. 
They showed that $\mathcal{S}_{k\textsc{-Colorable}}(n,\epsilon) = \tilde{O}(k^2/\epsilon^{3})$ for all $k \ge 3$. This result provides a significant improvement in the testability of colorability in the polychromatic regime: it shows that $k$-colorability is testable with a sublinear sample complexity when $\epsilon = \delta/k$ for a constant $\delta > 0$ whenever $k = o(n^{1/5})$.

Alon and Krivelevich~\cite{AlonKrivelevich02} further improved the bounds for testing $k$-colorability. They showed that $\mathcal{S}_{k\textsc{-Colorable}}(n,\epsilon) = \tilde{O}(k/\epsilon^{2})$ for all $k \ge 3$. This result shows that $k$-colorability is testable with a sublinear sample complexity for all $k = o(n^{1/3})$.
Sohler~\cite{Sohler12} obtained another improvement in the low chromatic number regime, showing that  $\mathcal{S}_{k\textsc{-Colorable}}(n,\epsilon) = \tilde{O}(k^6/\epsilon)$ and, in particular, when $k$ is a constant the sample complexity of testing $k$-colorability is $\tilde{\Theta}(1/\epsilon)$. But in the polychromatic setting, the bound of Alon and Krivelevich remains stronger whenever $k = \omega(1)$ and $\epsilon = \omega(1/k^5)$.

Our next main result unifies and improves on both of the incomparable results of Alon and Krivelevich~\cite{AlonKrivelevich02} and Sohler~\cite{Sohler12}.

\begin{theorem}
\label{thm:k-col}
The \textsc{$k$-Colorable} property has sample complexity $\mathcal{S}_{k\textsc{-Colorable}}(n,\epsilon) = \tilde{O}(\frac{k}{\epsilon})$.
\end{theorem}

\Cref{thm:k-col} implies that sublinear sample complexity suffices to distinguish $k$-colorable graphs from graphs whose $k$-colorings all result in at least $\frac{\delta}{k}n^2$ monochromatic edges for all values of $k = o(\sqrt{n})$ when $\delta > 0$ is constant.

We do not know whether the bound in \autoref{thm:k-col} is optimal. The best lower bound for testing $k$-colorability in the polychromatic regime is $\tilde{\Omega}(1/\epsilon)$~\cite{AlonKrivelevich02}. Since $\epsilon$-testing $k$-colorability is non-trivial only when $\epsilon < 1/k$, this means that there is a quadratic gap between our new upper bound and this lower bound for testing $k$-colorability for large values of $k$.

\subsection{Proof Overview and the Graph Container Method}
\label{sect:overview}

\Cref{thm:clique,thm:k-col} are both established by analyzing the natural tester that checks whether the sampled induced subgraph $G[S]$ has the property that we are testing or not. In the case of large cliques, this tester accepts with probability at least 1/2,\footnote{See the Remark at the end of \Cref{sect:proof} for a discussion related to this choice of acceptance probability.} and in the case of $k$-colorability, the tester always accepts when the graph has the property.
The challenge in the analysis, as is usually the case in property testing, is in showing that the tester rejects all graphs that are far from having the property with sufficiently large probability. We address this challenge with the use of the graph container method.

We first note that the graph container method is stated in terms of independent sets. In the case of testing $\rho n$-cliques, our application of the graph container method applies more naturally to testing $\rho n$-independent sets, which is equivalent by considering the complement of a graph.

Very briefly, the graph container method builds upon the observation that though a graph may contain a large number of independent sets, for every graph $G$ satisfying certain conditions (such as being $\epsilon$-far from having a $\rho n$ independent set) we can find a much smaller collection of sets of vertices that we call \emph{containers} such that (i) every independent set is a subset of at least one of the containers, (ii) the containers are small, and (iii) for each container $C$, the induced subgraph $G[C]$ is sparse.
The graph container method was originally introduced by Kleitman and Winston~\cite{kleitmanWinston1982} to bound the number of square-free graphs. The method has since been used to great effect, notably by Sapozhenko (see~\cite{Sapozhenko05} and the references therein), and was later extended to the hypergraph container method~\cite{baloghIndependentSetsHypergraphs2015,saxtonThomasonHypergraphContainers2015} that has seen even more wide use throughout combinatorics. The container method has also been recently applied to algorithmic settings for the problem of approximately counting independent sets in bipartite graphs~\cite{JenssenPP23} and the problem of obtaining faster exact algorithms for almost-regular graphs~\cite{Zamir23}.

Our use of the graph container method builds on Kleitman and Winston's original approach to it. The key component of the method as they introduced it is a simple greedy algorithm for identifying a small set of vertices of an independent set that form its \emph{fingerprint}. The fingerprint of an independent set $I$ in turn uniquely defines the \emph{container} that contains $I$.
Since the fingerprints are small sets, the collection of all possible fingerprints is relatively small.
Furthermore, the set of containers associated with all possible fingerprints cover all independent sets in the graph.
We describe the greedy algorithm in \Cref{sect:containers}.

Our main technical lemmas provide tight bounds on the size of the containers as a function of the size of the corresponding fingerprints in graphs that are $\epsilon$-far from the properties of interest.
For the problem of testing $\rho n$-independent sets, we prove the following graph container lemma.

\begin{restatable}[Graph Container Lemma I]{lemma}{GCLindepsetNew}
\label{lem:GCL-indepsetNew} 
Let $G = (V,E)$ be a graph on $n$ vertices that is $\epsilon$-far from the \textsc{$\rho$-IndepSet} property.
Then there exists a collection $\mathcal{F} \subseteq 2^V$ of fingerprints and a container function $\mathcal{C} : \mathcal{F} \rightarrow 2^V$ that satisfy the following two properties.
\begin{enumerate}
\item For every independent set $I$ in $G$ of cardinality $|I| \geq \frac{8\rho^2 \ln(\frac{2\rho}{\epsilon})}{\epsilon}$, there is a fingerprint $F \in \mathcal{F}$ such that $F \subseteq I \subseteq \mathcal{C}(F)$.
\item For every fingerprint $F \in \mathcal{F}$, the cardinality of $F$ and of its associated container $\mathcal{C}(F)$ are bounded by 
\[
	|F| \leq \frac{8\rho^2 \ln(\frac{2\rho}{\epsilon})}{\epsilon} \qquad \mbox{and} \qquad |\mathcal{C}(F)| \leq \left( \rho - |F| \cdot \frac{\epsilon}{8 \rho \ln(\frac{2\rho}{\epsilon})} \right) n.
\]
\end{enumerate}
\end{restatable}

The container lemma is used to analyze the soundness of testers for $\rho n$-independent sets in the following way. 
Fix any graph $G$ that is $\epsilon$-far from having a $\rho n$-independent set and 
let $S$ be the set of $s$ vertices sampled by the tester.
The tester accepts if and only if there is a subset $I \subseteq S$ of size $\rho s$ that forms an independent set in $G$.
By the first property of \Cref{lem:GCL-indepsetNew}, such a set $I$ can be a subset of $S$ if and only if the fingerprint $F$ corresponding to $I$ is contained in $S$ and an additional $\rho s - |F|$ vertices from $\mathcal{C}(I)$ are also included in $S$.
Using the upper bounds on the cardinality of the containers from the second property of the lemma, we can show that the probability that this many vertices from $\mathcal{C}(I)$ are sampled in $S$ is extremely small.
So small, in fact, that we can apply a union bound over all potential fingerprints in $S$ to bound the overall probability that the sample includes a large independent set.
For all the details of the proof see \Cref{sect:testing-IS}.

For testing $k$-colorability the overall approach is the same but a different argument is required because, for graphs that are $\epsilon$-far from $k$-colorable, we wish to have a collection of fingerprints and associated containers to cover $k$-colorable subgraphs instead of independent sets.
Since the vertices of any $k$-colorable subgraph can be partitioned into $k$ independent sets, we find fingerprints for each of the $k$ independent sets and use this sequence of fingerprints to define a container.
The Graph Container Lemma that we need in this case is the following.

\begin{restatable}[Graph Container Lemma II]{lemma}{GCLkcolNew}
    \label{lem:GCL-kcolNew}
    Let $G$ be $\epsilon$-far from $k$-colorable.
    Then there exists a collection $\mathcal{F} \subseteq (2^V)^k$ of sequences of fingerprints and a container function $\mathcal{C} : \mathcal{F} \rightarrow 2^V$ that satisfy the following two properties.
    \begin{enumerate}
        \item For every set $U \subseteq V$ such that $G[U]$ is $k$-colorable and $|U| \geq 4k/\epsilon,$ there is a sequence of fingerprints $F=(T_1,\dots,T_k) \in \mathcal{F}$ such that $\bigcup_{i \in [k]} T_i \subseteq U \subseteq \mathcal{C}(F).$
        \item For every sequence of fingerprints $F=(T_1,\dots,T_k) \in \mathcal{F},$ $|T_i| \leq \frac{4}{\epsilon}$ for every $i \in [k],$ and \[|\mathcal{C}(F)| \leq \left( 1 - \frac{\epsilon}{4\ln (\frac 1\epsilon)}\cdot \max_{i \in [k]} |T_i| \right)n.\]
    \end{enumerate}
\end{restatable}

See \Cref{sect:colorable} for the proof of this lemma and all the details for its use in the analysis of $k$-colorability testing.

\subsection{Discussion}

\Cref{thm:clique,thm:k-col} suggest a number of intriguing open problems. We discuss three of them briefly.

\paragraph{Property testing with the container method.}
Our results demonstrate that the graph container method provides a useful tool to study the problems of testing cliques and colorability.
In subsequent work \cite{blais2024new}, we show that the graph container method can be used to study all 0--1 graph partition properties, as defined by Nakar and Ron~\cite{NakarRon2018}.
It would be interesting to see if the method is useful for testing other classes of graph properties as well.
For example, the class of 0--1 graph partition properties does not include the property of containing one, or multiple, large cliques, and so we suspect there is a larger class of graph partition properties that the container method can be used on.

It would also be interesting to see what applications the hypergraph container method~\cite{baloghIndependentSetsHypergraphs2015,saxtonThomasonHypergraphContainers2015} might have for property testing.
In the subsequent work mentioned above, we showed this method can be used to provide tight analysis of testers for satisfiability and various partition properties of hypergraphs.
However, one of the most exciting aspects of the hypergraph container method in combinatorics is that it has led to a number of results for problems in combinatorics that at first glance do not appear to relate directly to hypergraphs.
Can the hypergraph container method similarly yield a new analysis technique for obtaining new property testing results in other domains as well?

\paragraph{Query complexity.}
What exactly is the relation between the query complexity of (adaptive or non-adaptive) testers that are free to query any potential edge in a graph and the query complexity of the canonical tester that samples a set of vertices and queries all potential edges within the corresponding induced subgraph? This question has been studied extensively in the dense graph setting (see, e.g.,~\cite{AlonFKS00,GoldreichTrevisan03,BogdanovTrevisan04,BogdanovLi2010,GonenRon10,GoldreichRon11,GoldreichWigderson21} and discussions in \cite{goldreich17,bhattacharyyaYoshida22}) but remains open for many natural properties of graphs. 

As discussed above, \Cref{thm:clique} gives a partial answer to this question for the case of testing large cliques, since it shows that the number of edge queries performed by the canonical $\epsilon$-tester for the \textsc{$\rho$-Clique} property is (up to logarithmic factors) identical to the query complexity of the best adaptive testing algorithm for the same problem when $\epsilon = \Omega(\rho^2)$. 
But what is the optimal query complexity when $\epsilon$ is smaller?
In subsequent work \cite{blais2024new}, we show that it is possible to $\epsilon$-test the \textsc{$\rho$-Clique} property with $\tilde{O}(\rho^5/\epsilon^{7/2})$ edge queries, which is strictly better than the number of edge queries made by the canonical tester in \Cref{thm:clique} when $\epsilon=o(\rho^2),$ but the optimal query complexity still remains open.

For $k$-colorability, the analogous question has already been posed even in the setting where $k=2$: Bogdanov and Li~\cite{BogdanovLi2010} conjectured that it is possible to $\frac1n$-test $2$-colorability (that is, bipartiteness) with $o(n^2)$ queries.  \Cref{thm:k-col} does not shed light on this conjecture (or its analogues for larger values of $k$) but it is possible that \Cref{lem:GCL-kcol} or a similar variant could be helpful in this conjecture and other related ones.

\paragraph{Time complexity.}
~\hyperlink{thm:clique-alt}{Theorem 1'} only bounds the sample complexity of the Densest $k$-Subgraph problem, but the proof of this bound also immediately implies fairly strong results on the time complexity of the problem. Namely, it shows that distinguishing graphs with $k$-cliques from those whose $k$-subgraphs all have density at most $(1-\delta)$ for some constant $\delta > 0$ can be done by sampling $s = O( \frac{n}{\delta^2k} \ln^3( \frac{n}{\delta^2 k}))$ vertices and checking whether the induced subgraph on those vertices contains a clique of size $(k/n)s = O(\delta^{-2} \ln^3( \frac{n}{\delta^2k}))$. Even by using exhaustive search for the latter task, the resulting time complexity is quasipolynomial in $n$---or, more precisely, it is $n^{O(\ln^3 n)}$. Stronger time complexity bounds already exist for this problem~\cite{FeigeSeltser97}, but it would be interesting to see if the graph container method could be used to directly improve the time complexity bounds for various property testing problems and related approximation problems.

\section{The Container Method}
\label{sect:containers}
In this section we describe the greedy algorithm used to construct fingerprints and containers, and give some basic properties of the algorithm.

\subsection{Fingerprint and Container Procedure}
The notions of \emph{fingerprints} and \emph{containers} of independent sets in graphs are defined according to the procedure presented in \Cref{alg:FingerprintContainer}, which is a slight variant of the greedy algorithm that was (implicitly) originally introduced by Kleitman and Winston~\cite{kleitmanWinston1982}. The algorithm works as follows. The first fingerprint of an independent set $I$ is simply the vertex $v$ contained in $I$ with the highest degree in $G,$ where ties are broken based on an assumed ordering of the vertices $V.$ Since $I$ is an independent set, any neighbour of $v$ must not be contained in $I$ and, furthermore, by the fact that $v$ has the highest degree in $G$ of all vertices in $I,$  any vertex with higher degree in $G$ cannot be in $I.$ The associated container is constructed by removing all vertices from these two collections. The procedure is repeated until every vertex from $I$ has been selected to the fingerprint.

\begin{algorithm}[ht]
\caption{\sc{Fingerprint \& Container Generator}}
\label{alg:FingerprintContainer}
\smallskip
\KwIn{A graph $G = (V,E)$ and an independent set $I \subseteq V$ }
\medskip
Initialize $F_0 \gets \emptyset$ and $C_0 \gets V$\;

\For{$t=1,2,\dots,|I|$} {
    $v_{t} \gets$ the vertex in $I \setminus F_{t-1}$ with largest degree in $G[C_{t-1}]$\;

    \smallskip 
    \tcp{Add $v_t$ to the fingerprint}
    $F_t \gets F_{t-1} \cup \{v_t\}$\;
    
    \smallskip
    \tcp{Remove all the neighbours of $v_t$ from the container}
    $C_t \gets C_{t-1} \setminus \big\{ w \in C_{t-1} : (v_t,w) \in E\big\}$\;

    \smallskip
    \tcp{And remove all vertices with higher degree than $v_t$ in $G[C_{t-1}]$}
    $C_t \gets C_t \setminus \big\{ w \in C_{t-1} : \deg_{G[C_{t-1}]}(w) > \deg_{G[C_{t-1}]}(v_t) \big\}$\;
}
Return $F_1,\ldots,F_{|I|}$ and $C_1,\ldots,C_{|I|}$\;
\end{algorithm}

The greedy algorithm for generating fingerprints and containers was originally defined with an appropriate stopping condition so that each independent set generates a single fingerprint and container.
For our proofs, however, it will be more convenient to work directly with the sequence of fingerprints $F_1,F_2,\ldots,F_{|I|}$ and their corresponding containers $C_1,C_2,\ldots,C_{|I|}$ generated by the algorithm for the independent set $I$. Furthermore, it is also convenient to extend the definition with $C_t=F_t=I$ for all $t > |I|.$ We refer to $F_t$ and $C_t$ as the \emph{$t$-th fingerprint} and \emph{$t$-th container} of $I$, respectively.

When we consider the fingerprints or containers of multiple independent sets, we write $F_t(I)$ and $C_t(I)$ to denote the $t$-th fingerprint and container of the independent set $I$. (When the current context specifies a single independent set, we will continue to write only $F_t$ and $C_t$ to keep the notation lighter.)

\subsection{Basic Properties of Containers} 

By construction, the sequences of fingerprints and containers of an independent set satisfy a number of useful properties. For instance, we have
\begin{equation}
\label{eq:container-inclusion}
F_1 \subseteq F_2 \subseteq \cdots \subseteq F_{|I|} = I \subseteq C_{|I|} \subseteq C_{|I|-1} \subseteq \cdots \subseteq C_2 \subseteq C_1
\end{equation}
so that for each $t = 1,2,\ldots,|I|$, $F_t \subseteq I \subseteq C_t$.

The construction of the algorithm also guarantees that the $t$-th container of an independent set is the same as the $t$-th container of its $t$-th fingerprint.

\begin{proposition}
\label{prop:container-closure}
For any graph $G = (V,E)$, any independent set $I$ in $G$, and any $t,$ the fingerprint $F_t(I)$ and container $C_t(I)$ of $I$ satisfy 
\[
C_t\big( F_t(I) \big) = C_t(I).
\]
\end{proposition}

\begin{proof}
If $t> |I|,$ then $F_t(I)=I=C_t(I)$ and the equality holds.
Now consider the case that $t \leq |I|.$ 
If $v_1,\ldots,v_t$ are the vertices selected in the first $t$ rounds of the greedy algorithm on input $I$, then the algorithm selects the same vertices and forms the same container $C_t$ when provided with input $F_t(I)$ instead of $I$.
\end{proof}

Another basic property of the containers of any independent set is that we can bound their maximum degree in the following way.

\begin{proposition}
\label{prop:container-degree}
For any graph $G = (V,E)$ on $|V| = n$ vertices, any independent set $I$ in $G$, and any $t \ge 1,$ the $t$-th container of $I$ has maximum degree bounded by $\Delta(G[C_t]) \le n/t$.
\end{proposition}

\begin{proof}
If $t > |I|,$ then $C_t=I,$ and so the maximum degree is $0.$
Now, consider any $t \le |I|$.
Since the containers satisfy the containment $V = C_0 \supseteq C_1 \supseteq \dots \supseteq C_t$, 
\[
n \ge |C_0| - |C_t| = \sum_{i=1}^t \big(|C_{i-1}| - |C_i| \big) = 
\sum_{i=1}^t |C_{i-1} \setminus C_i|
\]
and there must be an index $j \in [t]$ for which $|C_{j-1} \setminus C_j| \leq n/t$. This means that the vertex $v_j$ chosen in the $j$th round of the greedy algorithm must have degree at most $n/t$ in the induced subgraph $G[C_{j-1}]$. All vertices in $C_{j-1}$ that had higher degree than $v_j$ in $G[C_{j-1}]$ were removed when constructing $C_j$, so all vertices in $C_j$ have degree at most $n/t$ in $G[C_j]$. Since $C_t \subseteq C_j$, the maximum degree of $G[C_t]$ is bounded above by $n/t$ as well.
\end{proof}

\section{Testing Cliques and Independent Sets}
\label{sect:testing-IS}
We complete the proof of \Cref{thm:clique} in this section. The proof applies to both the \textsc{$\rho$-Clique} and the \textsc{$\rho$-IndepSet} properties (by considering the complement of the tested graph) but the latter fits in more naturally with the graph container method, so in the rest of the section we discuss the result entirely in terms of independent sets.

\subsection{Container Shrinking Lemma}
\label{sect:containerShrinking}

The main tool in the proof of the Graph Container Lemma for testing independent sets, \Cref{lem:GCL-indepsetNew}, is a shrinking lemma that, informally speaking, says that many vertices are removed in each round of \Cref{alg:FingerprintContainer}.

To gain some intuition for this shrinking lemma, observe that if $|C_t| \geq \rho n,$ then $G[C_t]$ has edge density at least $\frac{2\epsilon}{\rho^2}.$
Then, it is not hard to show that at least roughly $\frac{\epsilon}{2\rho^2}|C_t|$ vertices are removed when constructing $C_{t+1},$ which means that $|C_{t+1}| \leq \left(1-\frac{\epsilon}{2\rho^2}\right)|C_t|.$
This is because in the $t$-th step of \Cref{alg:FingerprintContainer} either $v_t$ has degree larger than $\frac{\epsilon}{2\rho^2}|C_t|$ in $G[C_t],$ or all vertices with degree larger than $\frac{\epsilon}{2\rho^2}|C_t|$ in $G[C_t]$ are removed when constructing $C_{t+1}$, and there are at least $\frac{\epsilon}{2\rho^2}|C_t|$ such vertices. 

The following lemma shows that if $G[C_t]$ contains a large sparse subgraph, then the container, or more specifically the number of vertices in the container outside of the sparse subgraph, shrinks by much more than a factor of $(1-\frac{\epsilon}{2\rho^2}).$

\begin{lemma}[Container Shrinking Lemma]
\label{lem:containers-shrinking}
Let $G = (V,E)$ be a graph on $n$ vertices that is $\epsilon$-far from the \textsc{$\rho$-IndepSet} property,  let $I$ be an independent set in $G$, and let $t < |I|$ be an index for which the $t$-th container $C_{t}$ of $I$ has cardinality $|C_{t}| \ge \rho n$. For any $\alpha > 0$, if there exists a subset $C \subseteq C_{t+1}$ of $(\rho - \alpha) n$ vertices where $G[C]$ contains at most $\frac{\epsilon}{4} n^2$ edges, then
\[
\left|C_{t+1} \setminus C\right| \le \left( 1 - \frac{\epsilon}{4\rho \alpha} \right) \left| C_{t} \setminus C \right|.
\]
\end{lemma}

\begin{proof}
First recall that $C \subseteq C_{t+1} \subseteq C_t,$ and so it holds that $|C_{t+1} \setminus C| = |C_t\setminus C| - |C_t \setminus C_{t+1}|.$
To prove the lemma we will give a lower bound on $|C_t \setminus C_{t+1}|.$
First suppose that $C$ contains a vertex $v$ with degree at least $\frac{\epsilon}{4\rho \alpha}|C_{t} \setminus C|$ in the graph $G[C_{t}|.$
Then, when constructing $C_{t+1}$ from $C_t$, at least this number of vertices are removed from $C_t$ since the vertex $v_{t+1}$ selected by the greedy algorithm has degree at least as large as that of $v$ in $G[C_t],$ otherwise $v$ would not have been in $C_{t+1}$ (and so therefore would not be in $C$).

Consider now the setting where all the vertices in $C$ have degree at most $\frac{\epsilon}{4\rho \alpha}|C_{t} \setminus C|$ in $G[C_t]$.
We claim that $G[C_{t} \setminus C]$ contains at least $\frac{\epsilon}{4 \rho \alpha} |C_{t} \setminus C|^2$ edges. This claim suffices to complete the proof of the lemma, since it implies that at least $\frac{\epsilon}{4 \rho \alpha} |C_{t} \setminus C|$ vertices in $C_t \setminus C$ have degree at least $\frac{\epsilon}{4 \rho \alpha} |C_{t} \setminus C|$ in $G[C_t \setminus C]$, and thus that at least this many vertices are removed from $C_t$ when constructing $C_{t+1}$ by \Cref{alg:FingerprintContainer}.

We now complete the proof of the claim that $G[C_{t} \setminus C]$ contains at least $\frac{\epsilon}{4 \rho \alpha} |C_{t} \setminus C|^2$ edges.
Let $R$ be chosen uniformly at random from all subsets of $C_t \setminus C$ of size $\alpha n$.
Since $G$ is $\epsilon$-far from having independent sets of size $\rho n$, the induced subgraph $G[R \cup C]$ contains at least $\epsilon n^2$ edges.
By the lemma hypothesis, $G[C]$ contains at most $\frac\epsilon4 n^2$ edges.
Finally, by the maximum degree of the vertices of $C$ in $G[C_t],$ the number of edges with one endpoint in $C$ and one endpoint in $C_t \setminus C$ is at most $|C| \cdot \frac{\epsilon}{4\rho \alpha}|C_{t} \setminus C| \leq \frac{\epsilon}{4\alpha}|C_{t} \setminus C| n.$
Each such edge is included in $G[R \cup C]$ with probability $|R| / |C_t \setminus C|,$ and so the expected number of edges between $C$ and $R$ is at most $\frac{\epsilon}{4\alpha} |R| n =\frac{\epsilon}{4}n^2$. 
Therefore, the expected number of edges in $G[R]$ is at least
\[
\frac{\epsilon}{2} n^2
= \frac{\epsilon}{2\alpha^2}|R|^2 
\geq \frac{\epsilon}{\alpha^2} \binom{|R|}{2}.
\]
In other words, the random graph $R$ chosen uniformly at random from the subgraphs of $C_t \setminus C$ of cardinality $\alpha n$ has expected density at least $\frac{\epsilon}{\alpha^2}$. This implies that $C_t \setminus C$ also has density at least $\frac{\epsilon}{\alpha^2}$. So it contains at least $\frac{\epsilon}{\alpha^2} \binom{|C_t \setminus C|}{2} \geq 
\frac{\epsilon}{\rho\alpha} \binom{|C_t \setminus C|}{2}
\ge \frac{\epsilon}{4\rho \alpha}|C_t \setminus C|^2$ edges, as we wanted to show.
\end{proof}

\subsection{Proof of Graph Container Lemma I}

In this section, we prove the Graph Container Lemma for testing independent sets.
We first use the Container Shrinking Lemma to prove the following alternative formulation of the graph container lemma for testing independent sets. 
We then show how this alternative graph container lemma implies \Cref{lem:GCL-indepsetNew}.

\begin{lemma}
\label{lem:GCL-indepset} 
Let $G = (V,E)$ be a graph on $n$ vertices that is $\epsilon$-far from the \textsc{$\rho$-IndepSet} property. For any independent set $I$ in $G$, there exists an index $t \le \frac{8 \rho^2}{\epsilon}  \ln(\frac{2\rho}\epsilon)$ such that the size of the $t$-th container of $I$ is bounded by
\begin{equation}
\label{eq:GCL-indepset}
|C_t| \le \left( \rho - t \cdot \frac{\epsilon}{8 \rho \ln(\frac{2\rho}{\epsilon})} \right) n
\end{equation}
and $G[C_t]$ contains at most $\epsilon n^2/4$ edges.
\end{lemma}

\begin{proof}
Let $C$ denote the first container in the sequence $C_1,C_2,\ldots$ that contains at most $\frac\epsilon4 n^2$ edges.
The existence of $C$ is guaranteed because $C_{|I|+1},$ which by definition is equal to $I,$ is an independent set and has no edges.
Define $\alpha$ such that $|C| = (\rho - \alpha) n$. The value of $\alpha$ is bounded by $\frac{\epsilon}{2\rho} \le \alpha \le \rho$, where the lower bound on $\alpha$ is due to the fact that $G$ is $\epsilon$-far from the \textsc{$\rho$-IndepSet} property.

Define $t^*$ to be the largest index for which $|C_{t^*}| \ge \rho n$.
First observe that $t^* \leq |I|$ because otherwise $C_{t^*}$ is an independent set, which contradicts the fact that $G$ is $\epsilon$-far from the \textsc{$\rho$-IndepSet} property.
Hence, by \Cref{lem:containers-shrinking} applied to all values of $t = 0,1,2, \ldots, t^*-1$,  
\[
|C_{t^*} \setminus C| \le \left(1 - \frac{\epsilon}{4\rho \alpha}\right)^{t^*}n
\]
and since $|C_{t^*} \setminus C| \ge \alpha n$, we conclude that $t^* \le \frac{4\rho \alpha}{\epsilon} \ln(1/\alpha) \le \frac{4\rho \alpha}{\epsilon} \ln (2\rho/\epsilon)$.

Consider now values $t > t^*$.
When $G[C_t]$ contains more than $\frac{\epsilon}{4} n^2$ edges, then the vertices in $C_t$ have average degree at least $\frac{\epsilon}{4\rho} n$ in this graph. 
This means that at least $\frac{\epsilon}{4\rho}n$ vertices in $C_t$ have degree at least $\frac{\epsilon}{4\rho} n$ in $G[C_t]$, so at least $\frac{\epsilon}{4\rho}n$ vertices are removed in the $t$-th iteration of \Cref{alg:FingerprintContainer}. 
Since $|C_{t^*+1} \setminus C| \le \alpha n$, there can be at most $\frac{4\rho \alpha}{\epsilon}$ such iterations before we reach $C.$
Hence, in total there are at most $\frac{4\rho \alpha}{\epsilon} \ln(2\rho/\epsilon)  + \frac{4\rho \alpha}{\epsilon} \le \frac{8\rho \alpha}{\epsilon} \ln (2\rho/\epsilon)$ iterations before we reach $C,$ and rearranging for $\alpha$ shows that the container $C$ satisfies the conclusion of the lemma. 
\end{proof}

\Cref{lem:GCL-indepsetNew}, restated below, follows almost immediately from \Cref{lem:GCL-indepset}.

\GCLindepsetNew*

\begin{proof}
For any independent set $I$ in $G$, let $t_I$ denote an index that satisfies the container cardinality bound in \Cref{lem:GCL-indepset}. 
    Let \[ \mathcal{F}=\biggl\{F_{t_I}(I) : I \textrm{ is an independent set in } G \textrm{ and } |I| \geq \frac{8 \rho^2}{\epsilon}  \ln\left(\frac{2\rho}\epsilon\right)\biggr\},\] and let the container function $\mathcal{C}: \mathcal{F} \rightarrow 2^V$ be the function that takes a fingerprint $F \in \mathcal{F}$ and returns $C_{|F|}(F),$ where $C_t(\cdot)$ and $F_t(\cdot)$ are the functions defined by \Cref{alg:FingerprintContainer}, as discussed following the algorithim.
    First observe that for any independent set $I$ we have that $|F_{t_I} (I)| \leq t_I.$
    Hence, by the basic properties of containers and \Cref{prop:container-closure}, we have that $F_{t_I}(I) \subseteq I \subseteq C_{t_I}(I)$ and $C_{t_I}(I) \subseteq C_{|F_{t_I}(I)|}(I) = C_{|F_{t_I}(I)|}(F_{t_I}(I))=\mathcal{C}(F_{t_I}(I)),$ which proves the first conclusion of the lemma.
    
    We now show the desired upper bounds on the size of every fingerprint $F \in \mathcal{F}$ and associated container $\mathcal{C}(F).$
    Since $|F_{t_I} (I)| \leq t_I$ holds for any independent set $I,$ then by the upper bound of $t_I \leq \frac{8 \rho^2}{\epsilon}  \ln\left(\frac{2\rho}\epsilon\right)$ from \Cref{lem:GCL-indepset} we get the upper bound on $|F_{t_I}(I)|.$
    Furthermore, for any independent set $I$ with $|I| \geq \frac{8 \rho^2}{\epsilon}  \ln\left(\frac{2\rho}\epsilon\right),$ observe that $|F_{t_I}(I)| = t_I.$
    This implies that \[\mathcal{C}(F_{t_I}(I)) = C_{|F_{t_I}(I)|}(F_{t_I}(I))=C_{t_I}(F_{t_I}(I))=C_{t_I}(I),\] where the third equality uses \Cref{prop:container-closure}.
    Hence, by the upper bound on $|C_{t_I}(I)|$ in \Cref{lem:GCL-indepset}, we get the desired upper bound on $|\mathcal{C}(F_{t_I}(I))|.$
\end{proof}

\begin{remark}
Since \Cref{lem:GCL-indepset} applies for any independent set it is possible to prove \Cref{lem:GCL-indepsetNew} so that the first conclusion of the lemma holds for any independent set, instead of just independent sets of size at least $\frac{8 \rho^2}{\epsilon}  \ln\left(\frac{2\rho}\epsilon\right).$ However this requires a slight variation on the definition of a fingerprint, and for \Cref{thm:clique} we only need that \Cref{lem:GCL-indepsetNew} applies for independent sets of size at least roughly $\frac{\rho^2 \ln^3(1/\epsilon)}{\epsilon}.$
\end{remark}

\subsection{Proof of \autoref{thm:clique}}\label{sect:proof}
We are now ready to use \Cref{lem:GCL-indepsetNew} to complete the proof of \Cref{thm:clique}. The proof of the theorem also makes use of the following form of Chernoff's bound for hypergeometric distributions.

\begin{lemma}[Chernoff's Bound]
\label{lem:chernoff}
Let $X$ be drawn from the hypergeometric distribution $H(N,K,m)$ where $m$ elements are drawn without replacement from a population of $N$ elements, $K$ of which are marked, and $X$ represents the number of marked elements that were drawn.
Then for any $\vartheta \geq \E[X]$,
\[
\Pr\big[X \geq \vartheta\big] \leq \exp\left(-\frac{(\vartheta-\E[X])^2}{\vartheta+\E[X]}\right).
\]
\end{lemma}

\begin{proof}
A standard multiplicative form of Chernoff's bound states that for all $\delta > 0$, the random variable $X$ satisfies $\Pr\big[X \geq (1+\delta)\E[X]\big] \leq \exp\left(-\frac{\delta^2\E[X]}{2+\delta}\right)$. (This standard bound is usually stated for sums of independent variables, but the same bound also applies to hypergeometric random variables as well. See, e.g.,~\cite{Mulzer2018}). Using the identity $\vartheta=\left(1+\frac{\vartheta-\E[X]}{\E[X]}\right)\E[X]$, we can apply this inequality with the parameter $\delta=\frac{\vartheta-\E[X]}{\E[X]}$ and simplify.
\end{proof}

We are now ready to prove \autoref{thm:clique}, restated below in its precise form (with the polylogarithmic term) and for the \textsc{$\rho$-IndepSet} property. 

\newtheorem*{clique-thm-precise}{Theorem 1}
\begin{clique-thm-precise}[Precise formulation]
The sample complexity of the \textsc{$\rho$-IndepSet} property is 
\[
\mathcal{S}_{\rho\textsc{-IndepSet}}(n,\epsilon) = O\left(\frac{\rho^3}{\epsilon^2} \ln^3\left(\frac 1\epsilon \right)\right).
\]
\end{clique-thm-precise}

\begin{proof}
Let $S$ be a random set of $s = c \cdot \frac{\rho^3}{\epsilon^2} \ln^3\left(\frac 1\epsilon \right)$ vertices drawn uniformly at random from $V$ without replacement, where $c$ is a large enough constant.
Note that for clarity of presentation, in the rest of the proof we ignore all integer rounding issues as they do not affect the asymptotics of the final result.

If $G$ contains a $\rho n$ independent set, then $S$ contains at least $\rho s$ vertices from this independent set with probability at least $\frac12$, since the number of such vertices follows a hypergeometric distribution, and the median of this distribution is at least $\rho s$~\cite{Neumann70,KaasBurhman1980}.

In the remainder of the proof we upper bound the probability that $S$ contains a $\rho s$-independent set when $G$ is $\epsilon$-far from containing a $\rho n$-independent set.
We consider the collection of fingerprints $\mathcal{F}$ and the container function $\mathcal{C}: \mathcal{F} \rightarrow 2^V$ from \Cref{lem:GCL-indepsetNew} which states that for every independent set $I$ in $G$ with $|I| > \frac{8 \rho^2}{\epsilon}  \ln\left(\frac{2\rho}\epsilon\right)$ there exists $F \in \mathcal{F}$ such that $F \subseteq I \subseteq \mathcal{C}(F).$
Hence, $S$ contains an independent set of size $\rho s > \frac{8 \rho^2}{\epsilon}  \ln\left(\frac{2\rho}\epsilon\right)$ only when there exists $F \in \mathcal{F}$ such that $S$ contains $F$ \textbf{and} $S$ contains at least $\rho s$ vertices from $\mathcal{C}(F).$
In other words, the probability that $S$ contains a $\rho s$ independent set is at most
\[\sum_{F \in \mathcal{F}} \Pr[F \subseteq S] \cdot \Pr\left[\left|\mathcal{C}(F) \cap S\right| \geq \rho s ~\mid~ F \subseteq S\right].\]

For a specific $F \in \mathcal{F},$ the probability that $S$ contains $F$ is at most
\[\frac{\binom{n-|F|}{s-|F|}}{\binom{n}{s}} = \frac{\binom{s}{|F|}}{\binom{n}{|F|}},\]
because there are at most $\binom{n}{s}$ ways to choose $S,$ and $\binom{n-|F|}{s-|F|}$ ways to choose an $S$ that includes $F.$
The probability that $S$ contains at least $\rho s$ vertices from $\mathcal{C}(F),$ conditioned on the fact that $S$ contains $F,$ can be upper bounded using Chernoff's bound as explained next.

Fix any fingerprint $F$ and set $t = |F|$. 
Let $S$ be a sample of $s$ vertices that contains $F$.
Let $X$ denote the number of vertices among the other $s-t$ sampled vertices that belong to $\mathcal{C}(F)$. 
By the upper bound on the cardinality of $\mathcal{C}(F)$ from the second property of \Cref{lem:GCL-indepsetNew}, the expected value of $X$ is bounded above by
\[
\E[X] \leq \left(\rho-\frac{t\epsilon}{8\rho\ln(\frac{2\rho}{\epsilon})}\right)(s-t) <  \rho s - \frac{t\epsilon s}{8\rho\ln(\frac{2\rho}{\epsilon})}
\le \rho s - t - \frac{t\epsilon s}{16 \rho\ln(\frac{2\rho}{\epsilon})},
\]
where the last inequality uses the fact that $s = c \cdot \frac{\rho^3}{\epsilon^2} \ln^3(\frac{1}{\epsilon})$ for a large enough constant $c$ and that the problem is non-trivial only when $\epsilon < \rho^2.$

By Chernoff's bound in \Cref{lem:chernoff}, the probability that we draw at least $\rho s - t$ vertices from $\mathcal{C}(F)$ in the additional $s - t$ vertices drawn to form $S$ is
\[
\Pr[ X \ge \rho s - t] 
\le \exp\left(\frac{-(\rho s-t-\E[X])^2}{\rho s - t+\E[X]}\right)
\leq \exp\left(\frac{-\left(\frac{t \epsilon s}{16 \rho \ln(\frac{2\rho}{\epsilon})}\right)^2}{2\rho s}\right)
\leq \exp\left(-\frac{t^2\epsilon^2s}{512 \rho^3\ln^2(\frac{2\rho}{\epsilon})} \right).
\]
Therefore, by applying a union bound over all values $t \le \frac{8\rho^2 \ln(\frac{2\rho}{\epsilon})}{\epsilon}$ and all possible choices of $F \in \mathcal{F}$ of size $t$, the probability that $S$ contains an independent set of size $\rho s$ is at most
\begin{align*}
\sum_{t} \sum_{\substack{F \in \mathcal{F} \\ |F|=t}} \frac{\binom{s}{|F|}}{\binom{n}{|F|}} \cdot \Pr\big[|\mathcal{C}(F) \cap S| \geq \rho s ~\big|~ F \subseteq S\big] 
\le & \sum_t \binom{s}{t} \exp\left(-\frac{t^2\epsilon^2s}{512 \rho^3\ln^2(\frac{2\rho}{\epsilon})} \right) \\
\le & \sum_t \exp\left( t \ln s - \frac{t^2\epsilon^2s}{512 \rho^3\ln^2(\frac{2\rho}{\epsilon})} \right) < \frac13
\end{align*}
where the second inequality uses the upper bound $\binom{s}{t} \le s^t,$ and the last inequality again uses the fact that $s = c \cdot \frac{\rho^3}{\epsilon^2} \ln^3(\frac{1}{\epsilon})$ for a large enough constant $c$. Hence, the probability that $S$ contains an independent set of size at least $\rho s$ is less than $1/3.$
\end{proof}

\begin{remark}
The completeness guarantee in the proof of \Cref{thm:clique} states that the algorithm correctly accepts graphs that have a $\rho n$-independent set with probability at least $\frac12$, instead of the usual $\frac23$ bound typically used in property testing. Since the soundness guarantee is that the algorithm accepts graphs that are $\epsilon$-far from the \textsc{$\rho$-IndepSet} property with probability at most $\frac13$, these guarantees still enable us to apply standard error reduction techniques to obtain algorithms that err with probability at most $\gamma$ with a $\ln(1/\gamma)$ multiplicative factor. In particular, it shows that there is also a standard tester with completeness guarantee $\frac23$ with the same asymptotic sample complexity.

Alternatively, it is also possible to obtain the $\frac23$ completeness guarantee directly by changing the algorithm slightly to check if the induced subgraph has the \textsc{$(\rho - \tau)$-IndepSet} property for some appropriate gap parameter $\tau := \tau(\rho,\epsilon)$. 
\end{remark}

\section{Testing k-Colorability}
\label{sect:colorable}
In this section we prove \autoref{thm:k-col}. 
We start by proving the Graph Container Lemma for $k$-colorability in \Cref{sect:col-GCL} and complete the proof of the theorem in \Cref{sect:col-proof}.

\subsection{Proof of Graph Container Lemma II}
\label{sect:col-GCL}
In this section, we prove \Cref{lem:GCL-kcolNew}, the main lemma that gives a collection of fingerprints and (small) associated containers for $k$-colorable subgraphs of graphs that are $\epsilon$-far from $k$-colorable.
As in the last section, we first prove an alternative formulation of the lemma, \Cref{lem:GCL-kcol}, and then show how it implies \Cref{lem:GCL-kcolNew}.

\begin{lemma}
\label{lem:GCL-kcol}
Let $G$ be $\epsilon$-far from $k$-colorable, and let $I_1,\dots,I_k$ be independent sets in $G.$ Then there exists $t \le \frac4\epsilon$ such that
\begin{equation}
\label{eq:GCL-kcol}
\left|\bigcup_{i=1}^k C_t(I_i)\right| \leq \left( 1 - t \cdot \frac{\epsilon}{4\ln (\frac 1\epsilon)} \right)n. 
\end{equation}
\end{lemma}

\begin{proof}
Assume on the contrary that for all $t \le 4/\epsilon$, 
\begin{equation}
\label{eqn:containers-union}
\left|\bigcup_{i=1}^k C_t(I_i)\right| > \left( 1 - t \cdot \frac{\epsilon}{4\ln (1/\epsilon)} \right)n.
\end{equation}
We will obtain a contradiction by constructing a $k$-partition $V_1,\dots,V_k$ of the vertices which shows that $G$ is not $\epsilon$-far from $k$-colorable.
    
Initialize the sets $V_1=C_{4/\epsilon}(I_1),\dots,V_k=C_{4/\epsilon}(I_k).$ Remove vertices from $V_1,\dots,V_k$ until every vertex appears in one set so that $V_1,\dots,V_k$ forms a partition of $\bigcup_{i=1}^k C_{4/\epsilon}(I_i).$ By \Cref{prop:container-degree}, each vertex has degree at most $\epsilon n/4$ in its set, and so the sum of the number of edges contained in the sets $V_1, \dots, V_k$ is at most $\frac{\epsilon n^2}{4}.$

    In the remainder of the proof, we allocate each of the remaining vertices in $V \setminus \bigcup_{i \in [k]}V_i= V \setminus \bigcup_{i \in [k]} C_{4/\epsilon}(I_i)$ to one of the parts and argue that the number of additional edges within the parts is small.
    
    For every $t = 1, \dots, \frac{4}{\epsilon},$ let $A_t$ be the set of vertices that are contained in at least one container at the $(t-1)$-th step, but not contained in any container at step $t$ step of the container procedure. In other words, let \[A_t = \left\{v \in V : v \in \bigcup_{i \in [k]} C_{t-1}(I_i) \textrm{ and } v \notin \bigcup_{i \in [k]} C_{t}(I_i)\right\}.\]
    
    Observe that $A_1,\dots,A_{4/\epsilon}$ partitions the vertices in $V \setminus \bigcup_{i \in [k]}V_i,$ and so to complete the partition of $G$ we add the vertices of $A_1, \dots, A_{4/\epsilon}$ to $V_1,\dots,V_k$ in the following way: for each $t=1,\dots, \frac{4}{\epsilon}$ and $v \in A_t,$ select an $i$ such that $v \in C_{t-1}(I_i),$ and add $v$ to set $V_i.$

    In order to count the new edges, consider adding the vertices from $A_1,\dots,A_{4/\epsilon}$ in reverse order (starting at vertices in $A_{4/\epsilon}$ and going to $A_1$). In this way, when $v \in A_t$ is added to $V_i,$ it holds that $V_i \subseteq C_{t-1}(I_i).$ Hence, by \Cref{prop:container-degree}, each $v \in A_t$ contributes at most $\frac{n}{t-1}$ new edges to $V_i$ (or $n$ edges if $t=1$) because $v$ is contained in $C_{t-1}(I_i).$
    
    Therefore, the sum of the edges in each of the sets $V_1, \dots, V_k$ can be upper bounded by
    \begin{equation}
        \label{eqn:sum-edges}
        \frac{\epsilon n^2}{4}+|A_1|\cdot n+\sum_{t = 2}^{4/\epsilon} |A_t| \frac{n}{(t-1)}.
    \end{equation}
    
    For any $t,$ $\bigcup_{\ell = 1}^t A_\ell = V \setminus \bigcup_{i \in [k]} C_t(I_i),$ and so \eqref{eqn:containers-union} implies that $\sum_{\ell = 1}^{t} |A_\ell| \leq \frac{t \epsilon n}{4\ln(1/\epsilon)}.$ In the sum \eqref{eqn:sum-edges}, the contribution from $|A_t|$ goes down as $t$ increases, and so the sum is maximized when $|A_t|=\frac{\epsilon n}{4\ln(1/\epsilon)}$ for all $t.$ Hence, the sum of the edges in each of the sets $V_1, \dots, V_k$ can be upper bounded by
    \[\frac{\epsilon n^2}{4}+\frac{\epsilon n^2}{4 \ln(1/\epsilon)} + \frac{\epsilon n^2}{4 \ln(1/\epsilon)}\sum_{t=2}^{4/\epsilon} \frac{1}{t-1} \leq \frac{\epsilon n^2}{4}+ \frac{\epsilon n^2}{4 \ln(1/\epsilon)}\left(\ln(4/\epsilon)+2\right) < \epsilon n^2,\]
    where the first inequality uses the upper bound $H(m) \le \ln(m) + 1$ on the harmonic series and the second inequality uses the bound $\epsilon < e^{-2}$ that we can apply without loss of generality. This is a contradiction with the fact that $G$ is $\epsilon$-far from $k$-colorable.
\end{proof}

We are now ready to complete the proof of \Cref{lem:GCL-kcolNew}, restated below.

\GCLkcolNew*
\begin{proof}
    For any $U \subseteq V$ such that $G[U]$ is $k$-colorable, let $\mathcal{I}(U)=(I_1,\dots,I_k)$ be an arbitrary partitioning of $U$ into $k$ independent sets (since $G[U]$ is $k$-colorable at least one such partitioning exists) and let $t_{\mathcal{I}(U)}$ denote the index given in \Cref{lem:GCL-kcol} with the desired container bounds.
    Let \[\mathcal{F}=\{(F_{t_{\mathcal{I}(U)}}(I_1),\dots,F_{t_{\mathcal{I}(U)}}(I_k)) : U\subseteq V \textrm{ is $k$-colorable and } |U| \geq 4k/\epsilon \},\] and let the container function $\mathcal{C}: \mathcal{F} \rightarrow 2^V$ be the function that takes a sequence of fingerprints $F = (T_1,\dots,T_k) \in \mathcal{F}$ and returns $\bigcup_{i \in [k]} C_{\max_{i \in [k]}|T_i|}(T_i),$ where $C_t(\cdot)$ and $F_t(\cdot)$ are the functions defined in \Cref{alg:FingerprintContainer}.
    First observe that for any $U \subseteq V$ such that $G[U]$ is $k$-colorable, $\max_{i \in [k]} |F_{t_{\mathcal{I}(U)}}(I_i)| \leq t_{\mathcal{I}(U)}.$
    Hence, by the basic properties of containers and \Cref{prop:container-closure}, we have that
    \[\bigcup_{i \in [k]} F_{t_{\mathcal{I}(U)}}(I_i) \subseteq U \subseteq \bigcup_{i \in [k]} C_{t_{\mathcal{I}(U)}}(I_i) = \bigcup_{i \in [k]} C_{t_{\mathcal{I}(U)}}(F_{t_{\mathcal{I}(U)}}(I_i)) \subseteq \mathcal{C}(F_{t_{\mathcal{I}(U)}}(I_1),\dots,F_{t_{\mathcal{I}(U)}}(I_k)),\]
    which proves the first conclusion of the lemma.

    We now show the desired upper bounds on the size of every $|T_i|$ for each $F=(T_1,\dots,T_k) \in \mathcal{F}$ and the size of the associated container $\mathcal{C}(F).$
    For any $U \subseteq V$ such that $G[U]$ is $k$-colorable and $|U| \geq 4k/\epsilon,$ we have that, for every $i \in [k],$ $|F_{t_{\mathcal{I}(U)}}(I_i)| \leq t_{\mathcal{I}(U)},$ and by the upper bound on $t_{\mathcal{I}(U)}$ from \Cref{lem:GCL-kcol} we get the desired upper bound on $|F_{t_{\mathcal{I}(U)}}(I_i)|.$
    Furthermore, observe that $\max_{i \in [k]} |F_{t_{\mathcal{I}(U)}}(I_i)|=t_{\mathcal{I}(U)}$ because $\max_{i \in [k]} |I_i| \geq |U|/k \geq 4/\epsilon,$ and so
    \[
    \mathcal{C}(F_{t_{\mathcal{I}(U)}}(I_1),\dots,F_{t_{\mathcal{I}(U)}}(I_k)) = \bigcup_{i \in [k]} C_{\max_{i \in [k]}|F_{t_{\mathcal{I}(U)}}(I_i)|}(F_{t_{\mathcal{I}(U)}}(I_i)) = \bigcup_{i \in [k]} C_{t_{\mathcal{I}(U)}}(F_{t_{\mathcal{I}(U)}}(I_i)).
    \]
    By \Cref{prop:container-closure} this is equal to $\bigcup_{i \in [k]} C_{t_{\mathcal{I}(U)}}(I_i),$ and using the upper bound from \Cref{lem:GCL-kcol} completes the upper bound on the size of $\mathcal{C}(F_{t_{\mathcal{I}(U)}}(I_1),\dots,F_{t_{\mathcal{I}(U)}}(I_k)).$
\end{proof}

\subsection{Proof of \autoref{thm:k-col}}
\label{sect:col-proof}
We can now complete the proof of \autoref{thm:k-col}, restated below.

\newtheorem*{col-thm-precise}{Theorem 2}
\begin{col-thm-precise}[Precise formulation]
The sample complexity of the \textsc{$k$-Colorable} property is 
\[
\mathcal{S}_{k\textsc{-Colorable}}(n,\epsilon) = O\left(\frac{k}{\epsilon} \ln^2\left(\frac{1}{\epsilon}\right)\right).
\]
\end{col-thm-precise}

\begin{proof}
Let $S$ be a random set of $s = c\frac{k}{\epsilon} \ln^2\left(\frac{1}{\epsilon}\right)$ vertices drawn uniformly at random from $V$ without replacement, where $c$ is a large enough constant. The tester accepts a graph $G$ if and only if $G[S]$ is $k$-colorable. When $G$ is $k$-colorable, then $G[S]$ is also $k$-colorable so the tester always accepts. 
In the remainder of the proof, we upper bound the probability that $G[S]$ is $k$-colorable when $G$ is $\epsilon$-far from $k$-colorable.

Consider the collection of sequences of fingerprints $\mathcal{F}$ and the container function $\mathcal{C} \colon \mathcal{F} \rightarrow 2^V$ from \Cref{lem:GCL-kcolNew}.
For every set $U \subseteq V$ for which $G[U]$ is $k$-colorable and $|U| \geq 4k/\epsilon$, there exists $F=(T_1,\dots,T_k) \in \mathcal{F}$ such that $\bigcup_{i \in [k]} T_i \subseteq U \subseteq \mathcal{C}(F).$
Hence, $G[S]$ is $k$-colorable only when there exists $F=(T_1,\dots,T_k) \in \mathcal{F}$ such that $S \subseteq \mathcal{C}(F)$ and $S$ contains $\bigcup_{i \in [k]}T_i.$
In other words, the probability that $G[S]$ is $k$-colorable is at most
\[\sum_{F=(T_1,\dots,T_k) \in \mathcal{F}} \Pr\left[\bigcup_{i \in [k]} T_i \subseteq S\right] \cdot \Pr\left[S \subseteq \mathcal{C}(F) ~\bigg|~ \bigcup_{i \in [k]} T_i \subseteq S\right].\]

Fix $F=(T_1,\dots,T_k) \in \mathcal{F}$ and let $t = \big| \bigcup_{i \in [k]} T_i \big|$.
The probability that $S$ contains $\bigcup_{i \in [k]}T_i$ is at most
\[\frac{\binom{n-t}{s-t}}{\binom{n}{s}} = \frac{\binom{s}{t}}{\binom{n}{t}},\]
because there are at most $\binom{n}{s}$ ways to choose $S,$ and $\binom{n-t}{s-t}$ ways to choose an $S$ that includes the vertices in $F.$
By the upper bound on $|\mathcal{C}(F)|$ in \Cref{lem:GCL-kcolNew}, the probability that $S \subseteq \mathcal{C}(F)$ conditioned on the fact that $S$ contains $\bigcup_{i \in [k]}T_i,$ can be upper bounded by $ \left( 1 - \frac{\epsilon}{4\ln (1/\epsilon)} \cdot \max_{i \in [k]} |T_i|\right)^{s-t}$.
Hence, the probability above can be upper bounded by
\[
    \sum_{t=1}^{4k/\epsilon} ~\sum_{\substack{F=(T_1,\dots,T_k) \in \mathcal{F}\\ | \bigcup_{i \in [k]} T_i | = t}} \frac{\binom{s}{t}}{\binom{n}{t}} \cdot \left( 1 - \frac{\epsilon}{4\ln (1/\epsilon)} \cdot \max_{i \in [k]} |T_i|\right)^{s-t}
    \leq \sum_{t=1}^{4k/\epsilon} k^t \binom{s}{t} \cdot \left( 1 - \frac{\epsilon}{4\ln (1/\epsilon)} \cdot \frac{t}{k}\right)^{s-t},
\]
where the inequality uses that $\max_{i \in [k]} |T_i| \geq \frac{t}{k}$ and that for any $t$ there are at most $\binom{n}{t} \cdot k^t$ ways to choose a sequence of fingerprints $F = (T_1,\dots,T_k)$ with $\big| \bigcup_{i \in [k]} T_i \big| = t.$
Finally, using that $s > 2t$ and that $\binom{s}{t} \leq s^t,$ this can be upper bounded by
\[\sum_{t=1}^{4k/\epsilon} k^t s^t \cdot \left( 1 - \frac{\epsilon}{4\ln (1/\epsilon)} \cdot \frac{t}{k}\right)^{s/2} \le \sum_{t=1}^{4k/\epsilon} \exp\left(t \ln(s k) - \frac{\epsilon t s}{8k\ln (1/\epsilon)}\right) < 1/3,\]
where the inequality uses the facts that $s = c\frac{k}{\epsilon} \ln^2\left(\frac{1}{\epsilon}\right)$ for a large enough constant $c$ and that the problem is non-trivial only when $\epsilon < 1/k$ and $\ln(\frac k\epsilon) \le \ln\frac1{\epsilon^2} = 2\ln \frac 1\epsilon$ in this regime.
\end{proof}

\section*{Acknowledgements}
Thank you to the anonymous reviewers for helpful comments on the paper.

\bibliographystyle{alpha}
\bibliography{references}

\begin{thebibliography}{AFKS00}

\bibitem[AFKS00]{AlonFKS00}
Noga Alon, Eldar Fischer, Michael Krivelevich, and Mario Szegedy.
\newblock Efficient testing of large graphs.
\newblock {\em Combinatorica}, 20(4):451--476, 2000.

\bibitem[AK02]{AlonKrivelevich02}
Noga Alon and Michael Krivelevich.
\newblock Testing k-colorability.
\newblock {\em SIAM Journal on Discrete Mathematics}, 15(2):211--227, 2002.

\bibitem[BESS78]{BollobasESS78}
B{\' e}la Bollob{\'a}s, Paul Erd{\H o}s, Mikl{\' o}s Simonovits, and Endre
  Szemer{\'e}di.
\newblock Extremal graphs without large forbidden subgraphs.
\newblock {\em Annals of Discrete Mathematics}, 3:29--41, 1978.

\bibitem[BL10]{BogdanovLi2010}
Andrej Bogdanov and Fan Li.
\newblock A better tester for bipartiteness?
\newblock arXiv preprint 1011.0531, November 2010.

\bibitem[BMS15]{baloghIndependentSetsHypergraphs2015}
J{\'o}zsef Balogh, Robert Morris, and Wojciech Samotij.
\newblock Independent sets in hypergraphs.
\newblock {\em Journal of the American Mathematical Society}, 28(3):669--709,
  2015.

\bibitem[BS24]{blais2024new}
Eric Blais and Cameron Seth.
\newblock New graph and hypergraph container lemmas with applications in
  property testing.
\newblock In {\em Proceedings of the 56th Annual ACM SIGACT Symposium on Theory
  of Computing}, 2024.

\bibitem[BT04]{BogdanovTrevisan04}
Andrej Bogdanov and Luca Trevisan.
\newblock Lower bounds for testing bipartiteness in dense graphs.
\newblock In {\em Proceedings. of the 19th {IEEE} Annual Conference on
  Computational Complexity}, pages 75--81, 2004.

\bibitem[BY22]{bhattacharyyaYoshida22}
Arnab Bhattacharyya and Yuichi Yoshida.
\newblock {\em Property Testing: Problems and Techniques}.
\newblock Springer Nature, 2022.

\bibitem[FLS04]{feigeCliqueTesting2004}
Uriel Feige, Michael Langberg, and Gideon Schechtman.
\newblock Graphs with tiny vector chromatic numbers and huge chromatic numbers.
\newblock {\em SIAM Journal on Computing}, 33(6):1338--1368, 2004.

\bibitem[FS97]{FeigeSeltser97}
Uriel Feige and Michael Seltser.
\newblock On the densest $k$-subgraph problem.
\newblock Technical Report No. CS97-16, Weizmann Institute of Science, 1997.

\bibitem[GGR98]{goldreichPropertyTesting1998}
Oded Goldreich, Shafi Goldwasser, and Dana Ron.
\newblock Property testing and its connection to learning and approximation.
\newblock {\em Journal of the ACM}, 45(4):653--750, 1998.

\bibitem[Gol17]{goldreich17}
Oded Goldreich.
\newblock {\em Introduction to property testing}.
\newblock Cambridge University Press, 2017.

\bibitem[GR10]{GonenRon10}
Mira Gonen and Dana Ron.
\newblock On the benefits of adaptivity in property testing of dense graphs.
\newblock {\em Algorithmica}, 58(4):811--830, 2010.

\bibitem[GR11]{GoldreichRon11}
Oded Goldreich and Dana Ron.
\newblock Algorithmic aspects of property testing in the dense graphs model.
\newblock {\em SIAM Journal on Computing}, 40(2):376--445, 2011.

\bibitem[GT03]{GoldreichTrevisan03}
Oded Goldreich and Luca Trevisan.
\newblock Three theorems regarding testing graph properties.
\newblock {\em Random Structures \& Algorithms}, 23(1):23--57, 2003.

\bibitem[GW21]{GoldreichWigderson21}
Oded Goldreich and Avi Wigderson.
\newblock Non-adaptive vs adaptive queries in the dense graph testing model.
\newblock In {\em Proceedings of the 62nd {IEEE} Annual Symposium on
  Foundations of Computer Science}, pages 269--275, 2021.

\bibitem[HMP21]{HuleihelMP2021}
Wasim Huleihel, Arya Mazumdar, and Soumyabrata Pal.
\newblock Random subgraph detection using queries.
\newblock {\em arXiv preprint arXiv:2110.00744}, 2021.

\bibitem[JPP23]{JenssenPP23}
Matthew Jenssen, Will Perkins, and Aditya Potukuchi.
\newblock Approximately counting independent sets in bipartite graphs via graph
  containers.
\newblock {\em Random Structures \& Algorithms}, 2023.

\bibitem[KB80]{KaasBurhman1980}
Rob Kaas and Jan~M Buhrman.
\newblock Mean, median and mode in binomial distributions.
\newblock {\em Statistica Neerlandica}, 34(1):13--18, 1980.

\bibitem[KW82]{kleitmanWinston1982}
Daniel~J. Kleitman and Kenneth~J. Winston.
\newblock On the number of graphs without 4-cycles.
\newblock {\em Discrete Mathematics}, 41(2):167--172, 1982.

\bibitem[Man17]{Manurangsi2017}
Pasin Manurangsi.
\newblock Almost-polynomial ratio {ETH}-hardness of approximating densest
  $k$-subgraph.
\newblock In {\em Proceedings of the 49th Annual {ACM} {SIGACT} Symposium on
  Theory of Computing}, pages 954--961, June 2017.

\bibitem[Mul18]{Mulzer2018}
Wolfgang Mulzer.
\newblock Five proofs of {Chernoff}’s bound with applications.
\newblock {\em Bulletin of EATCS}, 1(124), 2018.

\bibitem[Neu70]{Neumann70}
Peter Neumann.
\newblock {\"U}ber den median einiger dikreter verteilungen und eine damit
  zusammenh{\"a}ngende monotone konvergenz.
\newblock {\em Wissenschaftliche Zeitschrift der Technischen Universit{\"a}t
  Dresden}, 19:29--33, 1970.

\bibitem[NR18]{NakarRon2018}
Yonatan Nakar and Dana Ron.
\newblock On the testability of graph partition properties.
\newblock In {\em Approximation, Randomization, and Combinatorial Optimization.
  Algorithms and Techniques ({APPROX/RANDOM} 2018)}, page~13, 2018.

\bibitem[RD85]{RodlDuke85}
Vojt{\v e}ch R{\"o}dl and Richard~A. Duke.
\newblock On graphs with small subgraphs of large chromatic number.
\newblock {\em Graphs and Combinatorics}, 1(1):91--96, 1985.

\bibitem[RS19]{RaczSchiffer2019}
Mikl{\'o}s~Z. R{\'a}cz and Benjamin Schiffer.
\newblock Finding a planted clique by adaptive probing.
\newblock arXiv preprint 1903.12050, 2019.

\bibitem[Sap05]{Sapozhenko05}
Alexander Sapozhenko.
\newblock Systems of containers and enumeration problems.
\newblock In {\em Stochastic Algorithms: Foundations and Applications (SAGA
  2005)}, pages 1--13. Springer, 2005.

\bibitem[Soh12]{Sohler12}
Christian Sohler.
\newblock Almost optimal canonical property testers for satisfiability.
\newblock In {\em Proceedings of the 53rd Annual {IEEE} Symposium on
  Foundations of Computer Science}, pages 541--550, 2012.

\bibitem[ST15]{saxtonThomasonHypergraphContainers2015}
David Saxton and Andrew Thomason.
\newblock Hypergraph containers.
\newblock {\em Inventiones mathematicae}, 201(3):925--992, 2015.

\bibitem[Zam23]{Zamir23}
Or~Zamir.
\newblock Algorithmic applications of hypergraph and partition containers.
\newblock In {\em Proceedings of the 55th Annual ACM Symposium on Theory of
  Computing}, 2023.

\end{thebibliography}

\end{document}